\documentclass{amsart}
\usepackage[
colorlinks=true,citecolor=blue,hypertexnames=false]{hyperref}
\theoremstyle{plain}
\newtheorem{prop}{Proposition}       
\newtheorem{thm}{Theorem}     
\newtheorem{fact}{Fact}

\newcommand{\bQ}{\mathbb{Q}}
\newcommand{\bC}{\mathbb{C}}

\numberwithin{equation}{section}

\usepackage{color,comment}

\title[Chebyshev-like polynomials]{Generating functions of\\ Chebyshev-like polynomials}
\author[A. Bostan]{Alin Bostan}
\author[B. Salvy]{Bruno Salvy}
\address{Algorithms Project, Inria Rocquencourt, France}
\thanks{Work of the first two authors was supported in part by the Microsoft Research-Inria Joint Centre.}
\email{firstname.lastname@inria.fr}
\author[K. Tran]{Khang Tran}
\address{Department of Mathematics, University of Illinois at Urbana-Champaign, Urbana, Illinois, USA}
\thanks{The last author was supported by National Science Foundation
grant DMS 08-38434 EMSW21-MCTP: Research Experience for Graduate Students.}
\email{khangdtran@gmail.com}


\begin{document}
\begin{abstract}
In this short note, we give simple proofs of several results and conjectures formulated by Stolarsky and Tran concerning generating functions of some families of Chebyshev-like polynomials.
\end{abstract}	
\maketitle


\section{Introduction}
Stolarsky~\cite{Stolarsky2002} observed that the discriminant of the product of the polynomials 
\[K_{s}(x,y)=(1+y)^{2s}+xy^{s}\quad\text{and}\quad
f_{m}(y)=(y^{2m+1}-1)/(y-1)\]
has a very nice factorization:
\begin{equation}\label{firstconj}
\Delta_{y}(K_{s}f_{m})=C_{m}^{(s)}x^{2s-1}(x+2^{2s})H_{m}^{(s)}(x)^{4},
\end{equation}
with $C_{m}^{(s)}$ a rational constant and $H_m^{(s)}(x)$ the following polynomial in~${\mathbb{Q}}[x]$:
\begin{equation}\label{defH}
H_{m}^{(s)}(x)=\prod_{k=1}^{m}\left(x+4^{s}\cos^{2s}\left(\frac{k\pi}{2m+1}\right)\right).
\end{equation}
This was obtained empirically for specific values of~$m$ and~$s$. Moreover, Stolarsky conjectured that the generating functions of the polynomials~$H_m^{(s)}$ for fixed positive integer $s$ are rational, with explicit empirical formulas for $s=1,2,3$.
In~\cite{Tran2009}, Tran proved Eq.~\eqref{firstconj} with the explicit value
$C_{m}^{(s)}=(-1)^{m}(2m+1)^{2m-1}s^{2s}$
and obtained the generating function of the sequences of polynomials $H_m^{(1)}(x)$ and $H_m^{(2)}(x)$. 

Indeed, for $s=1$, the polynomials $H_{m}^{(s)}(x)$ are related to the classical Chebyshev polynomials of the second kind $U_n(x)$. From the definition
\[ U_n(x) = 2^n \cdot \prod_{k=1}^n \left( x - \cos \frac{k \pi}{n+1}\right),\]
it follows that $U_{2m}(x) = (-1)^m \cdot H_m^{(1)}(-4x^2)$.
From there, a simple derivation gives the generating function of the polynomials~$H_m^{(1)}(x)$.
\begin{prop} \label{prop:2}  
  The generating function of the sequence $H_{m}^{(1)}(x)$ is rational:
\[
\sum_{m \geq 0} H_{m}^{(1)}(x) t^m = \frac{1-t}{(1-t)^2 - xt}.
\]
\end{prop}
\begin{proof}
Starting from the classical rational generating function
\begin{equation} \label{eq:1}
	\sum_{n \geq 0} U_n(x) t^n = \frac{1}{1-2xt + t^2}, 
\end{equation}
a generating function for the even part is readily obtained:
\[ \sum_{m \geq 0} U_{2m}(x) t^{2m} = \frac12 \left(\frac{1}{1-2xt+t^2} + \frac{1}{1+2xt+t^2} \right) = \frac{1+t^2}{(1+t^2)^2-4x^2t^2},\] and thus we deduce
\[ \sum_{m \geq 0} H_{m}^{(1)}(-4x^2) t^m = \sum_{m \geq 0} U_{2m}(x) (-t)^{m} = \frac{1-t}{(1-t)^2+4x^2t},\]
from which the conclusion follows. \end{proof}
In~\cite{Tran2009}, Proposition~\ref{prop:2} was proved in a different way, and the  generating function for the case~$s=2$ was shown to be rational. We generalize these results by proving:
\begin{thm} \label{prop:3}
  For any $s \geq 1$, the generating function $F_s(x,t)=\sum_{m \geq 0} H_{m}^{(s)}(x) t^m$ belongs to~$\bQ(x,t)$.
Its denominator has degree at most $2^s$ in $t$, and its numerator has degree at most $2^s-1$ in $t$.
\end{thm}
Moreover, we give an algorithm that computes these rational functions explicitly.                                             
In Section~\ref{sec-general}, we prove this theorem. 
In Section~\ref{sec:algorithm}, we comment on the computational aspects.
We conclude with a few further observations in Section~\ref{sec-final}.

\section{Rational Generating Series}\label{sec-general}
For convenience, we work with the polynomial $G_m^{(s)}(x) = (-1)^m  H_m^{(s)}(-x)$. It is monic, with roots the $s$-th powers of the roots of $G_m^{(1)}(x) = (-1)^m  H_m^{(1)}(-x)$.  To prove the theorem, it is clearly sufficient to show that the generating function $\sum_{m \geq 0} G_{m}^{(s)}(x) t^m$ is rational.            

\subsection{Roots}
Letting $\varepsilon_s$ be a primitive $s$-th root of unity, we have \[G_m^{(s)}(x^s) = \prod_{G_m^{(1)}(\alpha)=0} (x^s - \alpha^s) = 
\prod_{G_m^{(1)}(\alpha)=0} \left( \prod_{j=0}^{s-1} \bigl( x - \alpha/\varepsilon_s^j \bigr) \right),
\]        
which equals
\[\prod_{j=0}^{s-1} \left( \prod_{G_m^{(1)}(\alpha)=0}  \bigl( x  - \alpha/\varepsilon_s^j \bigr)  \right)
=
\prod_{j=0}^{s-1} \varepsilon_s^{-jm}\cdot\prod_{j=0}^{s-1} G_m^{(1)}(\varepsilon_s^j x) 
=     
(-1)^{m(s-1)}
\cdot 
\prod_{j=0}^{s-1} G_m^{(1)}(\varepsilon_s^j x).
\]   

\noindent Therefore, the polynomial $G_m^{(s)}$ can be expressed using solely~$G_m^{(1)}$ by 
\begin{equation}\label{eq:Gms}
	G_m^{(s)}(x^s) = (-1)^{m(s-1)}
\cdot \prod_{j=0}^{s-1} G_m^{(1)}(\varepsilon_s^j x),
\end{equation}
and in particular it is given by the resultant 
$	G_m^{(s)}(x) = \operatorname{Res}_y \bigl( G_m^{(1)}(y), x-y^s \bigr).$
\subsection{Polynomials, Recurrences and Hadamard Products}
A consequence of Proposition~\ref{prop:2} is that the sequence of polynomials~$G_m^{(1)}(x)$ satisfies the linear recurrence 
\begin{equation}\label{recGm} 
G_{m+2}^{(1)}(x)+(2-x) G_{m+1}^{(1)}(x)+ G_{m}^{(1)}(x)=0, \quad G_0(x) = 1, \; G_1(x) = x-1.
\end{equation}
The coefficients of this recurrence are \emph{constant}, in the sense that they do not depend on the index~$m$. The characteristic polynomial of the sequence~$G_m^{(1)}(x)$ is the reciprocal $t^2+(2-x)t+1$ of the denominator of its generating function.  

Conversely, any sequence satisfying a linear recurrence with constant coefficients admits a rational generating function and this applies in particular to the sequence of polynomials $G_m^{(1)}(\varepsilon_s^j x)$ for any $j \geq 0$.

The product of sequences~$u_n$ and~$v_n$ that are solutions of linear recurrences with constant coefficients satisfies a linear recurrence with constant coefficients again (see, e.g., \cite[\S2.4]{Poorten1989}, \cite[Prop.~4.2.5]{Stanley1986}, \cite[\S2, Ex.~5]{CerliencoMignottePiras1987}). Its generating function, called the \emph{Hadamard product} of those of~$u_n$ and~$v_n$, is therefore rational. 
{}From Eq.~\eqref{eq:Gms}, it follows that the generating function $\sum_{m \geq 0} G_m^{(s)}(x^s) t^m$ is also rational (\emph{a priori} belonging to $\bC(x,t)$). 

Moreover, the reciprocal of the denominator of a Hadamard product of rational series has for roots the pairwise products of those of the individual series. 
Thus, letting $\alpha_1(x)$, 
$\alpha_2(x)$ 
denote the roots of $(1+t)^2-xt$, the reciprocal of the denominator of the generating function of $(-1)^{m(s-1)} \cdot G_m^{(s)}(x^s)$ is the characteristic polynomial
\[P_s(x,t) =  \prod_{1\leq i_1,\ldots,i_s \leq 2} \Big(t - \alpha_{i_1}(x) \alpha_{i_2}(\varepsilon_s x) \cdots \alpha_{i_s}(\varepsilon_s^{s-1} x) \Big). \]               
Note that, since $\alpha_1 \cdot \alpha_2 = 1$, the polynomial $P_s$ is self-reciprocal with respect to~$t$.

We prove that the polynomial $P_s(x,t)$ belongs to $\mathbb{Q}[x^s,t]$ by showing that all the (Newton) powersums of the roots of $P_s(x,t)$ belong to $\mathbb{Q}[x^s]$. For any $\ell \in \mathbb{N}$, the $\ell$-th powersum 
is equal to the product 
\begin{equation}\label{Tell}
\left(\alpha_1^{\ell}(x) + \alpha_2^{\ell}(x)\right)
\left(\alpha_1^{\ell}(\varepsilon_s x) + \alpha_2^{\ell}(\varepsilon_s x)\right)
\cdots \left(\alpha_1^{\ell}(\varepsilon_s^{s-1} x) + \alpha_2^{\ell}(\varepsilon_s^{s-1} x)\right),
\end{equation}
and thus has the form $T_{\ell,s}(x) = Q_{\ell}(x) Q_{\ell}(\varepsilon_s x) \cdots Q_{\ell}(\varepsilon_s^{s-1} x)$ for some 
polynomial $Q_{\ell}(x) \in \mathbb{Q}[x]$. The 
polynomial $T_{\ell,s}(x)$ being left unchanged under replacing $x$ by $\varepsilon_s x$, it belongs to $\mathbb{Q}[x^s]$,
and so do all the coefficients of $P_s$. This can also be seen from the expression of~$T_{\ell,s}$ as a resultant: $T_{\ell,s}(x)=\operatorname{Res}_y(y^s-x^s,Q_\ell(y))$.

In conclusion, $G_m^{(s)}(x^s)$ satisfies a recurrence with coefficients that are polynomials in $\mathbb{Q}[x^s]$, thus the series $\sum_{m \geq 0} G_m^{(s)}(x^s) t^m$ is rational and belongs to $\mathbb{Q}(x^s,t)$, and thus $F_s(x,t)$ belongs to $\mathbb{Q}(x,t)$. The assertion on the degree in $t$ of (the denominator and numerator of) $F_s(x,t)$ follows from the form of $P_s(x,t)$.      This concludes the proof of Theorem~\ref{prop:3}.

\section{Algorithm}   \label{sec:algorithm}     
The proof of Theorem~\ref{prop:3} is actually effective, and therefore it can be used to generate, for specific values of $s$, the corresponding rational function $F_s$, in a systematic and unified way. 

\subsection{New Proof for the Case $s=2$}\label{sec-s=2}
We illustrate these ideas in the simplest case $s=2$, for which the computations can be done by hand.                

We compute the generating function of $k_m(x):=H_m^{(2)}(-x^2)=G_m^{(1)}(x)  G_m^{(1)}(-x)$ from which that of the sequence~$H_m^{(2)}(x)$ is easily deduced. By Proposition~\ref{prop:2}, the generating functions of the sequences $G_m^{(1)}(x)$ and $G_m^{(1)}(-x)$ have for denominators
\[g_1(t):=(1+t)^2-xt,\qquad g_2(t):=(1+t)^2+xt,\]
which are self-reciprocal. The polynomial whose roots are the products of the roots of~$g_1$ and~$g_2$ is obtained by a resultant computation:
\[\operatorname{Res}_u(g_1(u),u^2g_2(t/u))=(t-1)^4+x^2t(1+t)^2=1+(x^2-4)t+(2x^2+6)t^2+(x^2-4)t^3+t^4.\]

It follows that the sequence $k_m(x)$ satisfies the fourth order recurrence
\[ k_{m+4}+(x^2-4) k_{m+3}+(2x^2+6)k_{m+2}+(x^2-4)k_{m+1}+k_{m}=0.\] 
The initial conditions can be determined separately 
\[k_0 = 1, \; k_1
= 1-x^2, \; k_2 = 1-7x^2+x^4, \; k_3 = 1-26x^2+13x^4-x^6\]      
and used to compute the numerator of the generating function
of the sequence of polynomials $H_m^{(2)}(x)$.
In conclusion, we have just proven the identity 
\begin{equation}\label{gf2}
\sum_{m=0}^\infty{H_m^{(2)}(x)t^m}=\frac{(1-t)^3}{(t-1)^4-xt(t+1)^2}.
\end{equation}
This provides a simple alternative proof of~\cite[Prop.~4.1]{Tran2009}.

\subsection{General Algorithm}
The computation of Hadamard products extends to linear recurrences that can even have \emph{polynomial} coefficients. It is implemented in the Maple package {\tt gfun}~\cite{SalvyZimmermann1994}. This makes it possible to write a first implementation that produces the generating function of the~$H_m^{(s)}(x)$ for small~$s$. 

The computation can be made more efficient by using recent algorithms specific to linear recurrences with constant coefficients~\cite{BostanFlajoletSalvySchost2006} (or, equivalently, to rational functions). The central idea is to avoid the use of a general-purpose algorithm for the computation of bivariate resultants and recover the denominator from the (Newton) powersums of its roots, themselves obtained from the powersums of~$\alpha_1$ and~$\alpha_2$. Using very simple univariate resultants, computations in algebraic extensions can be avoided. 
This is summarized in the following algorithm that follows the steps of the constructive proof of Thm.~\ref{prop:3}, using efficient algorithmic tools.

\medskip
\noindent {\bf Input}: an integer $s\geq 2$.\\
\noindent {\bf Output}: the rational function $F_s(x,t)$.

\medskip
\begin{enumerate}
	\item Compute the powersums~$Q_\ell(x):=\alpha_1(x)^\ell+\alpha_2(x)^\ell\in\mathbb{Q}[x]$ for~$\ell=0,\dots,2^s$;
     	\item Compute the powersums~$T_{\ell,s}$ from Eq.~\eqref{Tell} by  	$T_{\ell,s}(x^{1/s})=\operatorname{Res}_y(y^s-x,Q_\ell(y))$ for~$\ell=0,\dots,2^s$;	
	\item Recover the polynomial $P_s(x^{1/s},t)$ from its powersums~$T_{\ell,s}(x^{1/s},t)$;
	\item Deduce the denominator $D_s(x,t)$ of $F_s(x,t)$  using 
	$P_s(x^{1/s},t) = D_s(-x,(-1)^s t)$;
	\item Compute $G_m^{(1)}(x)$ for~$m=0,\dots,2^s-1$ using the 2nd order recurrence~\eqref{recGm};
	\item Compute $G_m^{(s)}(x)=\operatorname{Res}_y(G_m^{(1)}(y),x-y^s)$ for~$m=0,\dots,2^s-1$;
	\item Compute the numerator $N_s(x,t):=D_s(x,t)\times\sum_{m=0}^{2^s-1}G_m^{(s)}(-x) (-t)^m \!\! \mod t^{2^s}$;
	\item {\bf Return} the rational function $F_s(x,t) = N_s(x,t)/D_s(x,t).$
\end{enumerate}
Steps~(1) and~(3) can be performed efficiently using the algorithms in~\cite{BostanFlajoletSalvySchost2006}.  
In Step~(4) we use the fact that $P_s(x,t)$ is the denominator of the generating series \[\sum_{m \geq 0} (-1)^{m(s-1)} \cdot G_m^{(s)}(x^s)t^m = \sum_{m \geq 0} H_m^{(s)}(-x^s) \bigl( (-1)^s t \bigr)^m = F_s(-x^s,(-1)^s t).\]
We give the complete (and remarkably short) Maple code in an Appendix.

\subsection{Special Cases}
For~$s=2$, the computation recovers~\eqref{gf2}.
For~$s=3$, it takes less than one hundredth of a second of computation on a personal laptop to prove the following result
\[  \sum_{m \geq 0} H_{m}^{(3)}(x) t^m  = \frac{(1-t) \left( (t-1)^6- xt^2 (t+3)(3t+1) \right)}{x^2 t^4-xt(t^4+14t^3+34t^2+14t+1)(t
-1)^2+(t-1)^8}, \]
that was conjectured by Stolarsky in~\cite{Stolarsky2002}.
 
Similarly, less than two hundredth of a second of computation is enough to discover and prove the following new result.
\begin{prop}
 The generating function $\sum_{m \geq 0} H_{m}^{(4)}(x) t^m$ is equal to
\[ \frac{(t-1) \left( x^2 t^4 A(t) - 2 x t^2 \, B(t) \, (t-1)^6 + (t-1)^{14} \right)}{x^3 t^5 (t+1)^2 \, (t-1)^4 + x^2 t^3  \, C(t) + x t  (t-1)^8 \, D(t)\, - (t-1)^{16}},    
 \]
where        
\begin{small}
\begin{alignat*}1 
A(t) & =  9t^6-46t^5-89t^4-260t^3-89t^2-46t+9, \\
B(t) & =  11t^4+128t^3+266t^2+128t+11,            \\            
C(t) & =         2t^{10}-13t^9+226t^8-300t^7-676t^6-2574t^5-676t^4-300t^3+226t^2-13t+2,\\
D(t) & =  t^6+60t^5+519t^4+1016t^3+519t^2+60t+1.
\end{alignat*} 
\end{small}
\end{prop}

The code given in the appendix makes it possible to compute all the rational functions $F_s(x,t)$ for $1\leq s\leq 7$ in less than a minute.   

\section{Further Remarks}\label{sec-final} 
While computing with the polynomials $H_m^{(s)}$ we experimentally discovered the following amusing facts.

\medskip
\begin{fact}\label{fact:0}
For all $s \geq 1$, we have:
\begin{alignat*}1 
H_0^{(s)}(x) & =  1, \\
H_1^{(s)}(x) & =  x+1, \\
H_2^{(s)}(x) & =  x^2+ L_{2s} \, x + 1, 
\end{alignat*}                    
where $L_{n}$ denotes the $n$-th element of the Lucas sequence defined by $L_0=2,L_1=1$ and $L_n = L_{n-1} + L_{n-2}$ for all $n \geq 2$.
\end{fact}

This is easy to prove.

\medskip
\begin{fact}\label{fact:1}
   $H_{m}^{(s)}(x)$ has only non-negative integer coefficients,
for all $m \geq 0$ and $s \geq 1$.
\end{fact}

Non-negativity is clear from~\eqref{defH}, while the integrality follows from the generating function when~$s=1$ and from the resultant representation of~$H_m^{(s)}$ for higher~$s$.
             
\smallskip
Fact~\ref{fact:1} suggests that the coefficients of the polynomials $H_{m}^{(s)}(x)$ could admit a nice combinatorial interpretation.    

For instance, the coefficient of $x^1$ in $H_m^{(s)}(x)$ is equal to the trace of the matrix $M^{2s}$, where $M=(a_{i,j})_{i,j=1}^m$ is the $m \times m$ matrix with $a_{i,j}=1$ for $i+j \leq m+1$ and $a_{i,j}=0$ otherwise.

\medskip 
\begin{fact}\label{fact:2}
   The rational function $F_s(x,t) = \sum_{m \geq 0} H_{m}^{(s)}(x) t^m$ writes $N_s(x,t)/D_s(x,t)$, where
$N_s(x,t)$ is a polynomial in $\mathbb{Q}[x,t]$ of degree $2^s-1$ in $t$ and  $B_{s-1} - 1$ in $x$,
$D_s(x,t)$ is a 
polynomial in $\mathbb{Q}[x,t]$ of degree $2^s$ in $t$ and  $B_{s-1}$ in~$x$, and where $B_n$ is the central binomial coefficient \[B_n = \binom{n}{\lfloor n/2 \rfloor}.\]
\end{fact}                 

Here is the sketch of a proof. We only prove upper bounds, but a finer analysis could lead to the exact degrees.

First, it is enough to show that the degree in $x$ of the polynomial $P_s(x,t)$ defined in the proof of Theorem~\ref{prop:3} is at most $s B_{s-1}.$ Now, instead of considering $P_s$, we study the simpler polynomial 
\[C_s(x,t) =  \prod_{1\leq i_1,\ldots,i_s \leq 2} \Big(t - \alpha_{i_1}(x) \alpha_{i_2} (x) \cdots \alpha_{i_s}(x) \Big). \]  
We will prove the bound $s B_{s-1}$ on the degree of $C_s$ in $x$. Adapting the argument to the case of $P_s$ is not difficult.

The starting point is that, when $x$ tends to infinity, $\alpha_1(x)$ grows like $x$, while $\alpha_2(x)$ grows like $x^{-1}$. Therefore, 
\[C_s(x,t) =  \prod_{k=0}^s \Big(t - \alpha_{1}^{s-k}(x)  \alpha_{2}^{k}(x) \Big)^{s \choose k} \]      
grows like  
$L_s(x,t) = \prod_{k=0}^s \Big(t - x^{s-k}  (x^{-1})^{k} \Big)^{s \choose k}.$
The degree in~$x$ of $C_s(x,t)$ is thus bounded by the degree in $x$ of the Laurent polynomial~$L_s(x,t)$, which is equal to 
\[ \deg_x(L_s) = \sum_{k=0}^{\lfloor s/2 \rfloor} {s \choose k} (s- 2k) = {s \choose {\lfloor s/2 \rfloor + 1}} \cdot \bigl( {\lfloor s/2 \rfloor + 1} \bigr) = s B_{s-1}.\]   

\section*{Appendix: Maple Code}
For completeness, we give a self-contained Maple implementation that produces the rational form of the generating function~$F_s(x,t)$.   

\bigskip

\begin{verbatim}
recipoly:=proc(p,t) expand(t^degree(p,t)*subs(t=1/t,p)) end:

newtonsums:=proc(p,t,ord)
local pol, dpol;
  pol:=p/lcoeff(p,t);
  pol:=recipoly(pol,t);
  dpol:=recipoly(diff(pol,t),t);
  map(expand,series(dpol/pol,t,ord))
end:

invnewtonsums:=proc(S,t,ord) series(exp(Int((coeff(S,t,0)-S)/t,t)),t,ord) end:

Fs:=proc(s,x,t) 
local H, G, N, Q, T, ell, y, Ps, Ds, G1, Gs, m, Ns; 
  H:=(1-t)/((1-t)^2-x*t); 
  G:=subs(x=-x,t=-t,H); 
  N:=2^s; 
  Q:=subs(x=y,newtonsums(denom(G),t,N+1)); 
  T:=series(add(resultant(y^s-x,coeff(Q,t,ell),y)*t^ell,ell=0..N),t,N+1); 
  Ps:=convert(invnewtonsums(T,t,N+1),polynom);  
  Ds:=subs(t=(-1)^s*t,x=-x,Ps); # denominator
  G1:=subs(x=y,series(G,t,N));      
  Gs:=series(add(resultant(coeff(G1,t,m),x-y^s,y)*t^m,m=0..N),t,N); 
  Ns:=convert(series(Ds*subs(x=-x,t=-t,Gs),t,N),polynom); # numerator 
  collect(Ns/Ds,x,factor) 
end:

\end{verbatim}

\bibliographystyle{plain}

\begin{thebibliography}{1}

\bibitem{BostanFlajoletSalvySchost2006}
Alin Bostan, Philippe Flajolet, Bruno Salvy, and {\'E}ric Schost.
\newblock Fast computation of special resultants.
\newblock {\em Journal of Symbolic Computation}, 41(1):1--29, January 2006.

\bibitem{CerliencoMignottePiras1987}
L.~Cerlienco, M.~Mignotte, and F.~Piras.
\newblock Suites r{\'e}currentes lin{\'e}aires. {P}ropri{\'e}t{\'e}s
  alg{\'e}briques et arithm{\'e}tiques.
\newblock {\em L'Enseignement Math{\'e}matique}, XXXIII:67--108, 1987.
\newblock Fascicule 1-2.

\bibitem{SalvyZimmermann1994}
Bruno Salvy and Paul Zimmermann.
\newblock Gfun: a {M}aple package for the manipulation of generating and
  holonomic functions in one variable.
\newblock {\em ACM Transactions on Mathematical Software}, 20(2):163--177,
  1994.

\bibitem{Stanley1986}
R.~P. Stanley.
\newblock {\em Enumerative Combinatorics}, volume~I.
\newblock Wadsworth \& Brooks/Cole, 1986.

\bibitem{Stolarsky2002}
Kenneth~B. Stolarsky.
\newblock Discriminants and divisibility for {C}hebyshev-like polynomials.
\newblock In {\em Number theory for the millennium, {III} ({U}rbana, {IL},
  2000)}, pages 243--252. A K Peters, Natick, MA, 2002.

\bibitem{Tran2009}
Khang Tran.
\newblock Discriminants of {C}hebyshev-like polynomials and their generating
  functions.
\newblock {\em Proceedings of the American Mathematical Society},
  137(10):3259--3269, 2009.

\bibitem{Poorten1989}
A.~J. van~der Poorten.
\newblock Some facts that should be better known, especially about rational
  functions.
\newblock In {\em Number theory and applications ({B}anff, {AB}, 1988)}, volume
  265 of {\em NATO Adv. Sci. Inst. Ser. C Math. Phys. Sci.}, pages 497--528.
  Kluwer Acad. Publ., Dordrecht, 1989.

\end{thebibliography}

\end{document}